\documentclass[a4paper,english,nostix]{birthday}
\usepackage[latin1]{inputenc}
\usepackage{amsmath,babel,amsthm,refcount}

\newtheorem{theorem}{Theorem}[section]
\newtheorem{lemma}[theorem]{Lemma}
\newtheorem{proposition}[theorem]{Proposition}

\theoremstyle{definition}

\theoremstyle{remark}

\newtheorem{remarks}[theorem]{Remarks}

\usepackage{enumerate}

\title{On the spectrum of leaky surfaces with a potential bias}
\author{Pavel Exner}

\contact[exner@ujf.cas.cz]{Department of Theoretical Physics, Nuclear Physics Institute, Czech Academy of Sciences, 25068 \v{R}e\v{z} near Prague, and Doppler Institute for Mathematical Physics and Applied
Mathematics, Czech Technical University, B\v{r}ehov\'{a}~7, 11519
Prague, Czechia}

\begin{document}

\noindent\emph{To my friend Helge Holden on the occasion of his 60th birthday.}

\begin{abstract}
We discuss operators of the type $H = -\Delta + V(x) - \alpha
\delta(x-\Sigma)$ with an attractive interaction, $\alpha>0$, in
$L^2(\mathbb{R}^3)$, where $\Sigma$ is an infinite surface,
asymptotically planar and smooth outside a compact, dividing the
space into two regions, of which one is supposed to be convex, and
$V$ is a potential bias being a positive constant $V_0$ in one of
the regions and zero in the other. We find the essential spectrum
and ask about the existence of the discrete one with a particular
attention to the critical case, $V_0=\alpha^2$. We show that
$\sigma_\mathrm{disc}(H)$ is then empty if the bias is supported in
the `exterior' region, while in the opposite case isolated
eigenvalues may exist.
\end{abstract}



\section{Introduction}

An anniversary is usually an opportunity to look back at the
achievements of the jubilee. In Helge's case the picture is
impressive as he contributed significantly to several different
areas of mathematical physics. Nevertheless, one of his works made a
much larger impact than any other, namely the monograph \cite{AGHH}
first published in 1988. It is a collective work but Helge's hand is
unmistakably present in the exposition, and I add that it makes me
proud to be a part of the second edition of this book.

This also motivates me to choose a problem from this area as a topic
of this paper. I am going to discuss operators of the type
\begin{equation} \label{hamiltonian}
H = -\Delta + V(x) - \alpha \delta(x-\Sigma)\,, \quad \alpha>0\,,
\end{equation}
in $L^2(\mathbb{R}^3)$, where the $\delta$-potential is supported by
an infinite surface $\Sigma$ dividing the space into two regions, of
which one is supposed to be convex, and $V$ is a potential bias
being a positive constant in one of the regions and zero in the
other. The question to be addressed concerns spectral properties of
such operators, in particular, how they depend on the geometry of
$\Sigma$. We will observe similarities with recent results obtained
in the two-dimensional analogue of the present problem \cite{EV16},
especially a peculiar asymmetry. On the other hand, however, it is
not likely to have these results extended to higher dimensions, cf.
a remark at the end of Section~\ref{s:cone}.

\section{Statement of the problem and the results}
\label{s:results}

Let us begin with formulating the problem described above in proper terms.

\subsection{Assumptions}

As we shall see, the geometry of $\Sigma$ will be decisive for
spectral properties of the operator \eqref{hamiltonian}. We focus
our attention on the following class of surfaces:
  \begin{enumerate}
  \item[(a)] $\Sigma$ is topologically equivalent to a plane dividing $\mathbb{R}^3$ into two regions such that one of them is convex. The trivial case of two halfspaces is excluded.
  \item[(b)] $\Sigma$ may contain at most finite families $\mathcal{C}=\{C_j\}$ of finite $C^1$ curves, which are either closed or have regular ends, and $\mathcal{P}=\{P_j\}$ of points such that outside the set $\mathcal{C} \cup \mathcal{P}$ the surface is $C^2$ smooth admitting a parametrization with a uniformly elliptic metric tensor.
    \end{enumerate}
To distinguish the two regions we shall refer to the convex one as
`interior' and denote it $\Omega_\mathrm{int}$, the other will be
`exterior', denoted as $\Omega_\mathrm{ext}$. The curve finiteness
in assumption (b) refers to the Hausdorff distance, i.e. to the
metric inherited from the ambient three-dimensional Euclidean space.
The assumption implies, in particular, that $\Sigma$ is $C^2$ smooth
outside a compact. Note also that the curves indicating the
non-smooth parts of $\Sigma$ may in general touch or cross; the
regular ends mean that they can be prolonged locally without losing
the $C^1$ property.

Since $\Sigma$ is by assumption topologically equivalent to a plane,
we can use an atlas consisting of a single chart, in other words, a
map $\Sigma:\,\mathbb{R}^2\to\mathbb{R}^3$ provided we accept the
licence that the fundamental forms and quantities derived from them
may not exist at the points of $\mathcal{C} \cup \mathcal{P}$,
nevertheless, geodesic distances remain well defined across these
singularities. Furthermore, by assumption (b) the principal
curvatures $k_1,k_2$ of $\Sigma$ are well defined outside a compact.
We will suppose that
  \begin{enumerate}
  \item[(c)] $\Sigma$ is \emph{asymptotically planar}, that is, the principal curvature vanish as the geodesic distance from a fixed point tends to infinity.
    \end{enumerate}
Equivalently one can require that both the Gauss and mean curvatures
given by $K=k_1k_2$ and $M=\frac12(k_1+k_2)$, respectively, vanish
asymptotically. We also assume that
  \begin{enumerate}
  \item[(d)] there is a $c>0$ such that $|\Sigma(s)- \Sigma(t)| \ge c|s-t|$ holds for any $s,t\in \mathbb{R}^2$.
    \end{enumerate}
This ensures, in particular, that there are no cusps at the points
of $\mathcal{C} \cup \mathcal{P}$ where $\Sigma$ is not smooth; in
view of assumption (a) such a constant must satisfy $c<1$.

Given a bounded potential $V$, one can demonstrate in the same way
as in \cite[Sec.~4]{BEKS94} that under the stated assumptions the
quadratic form $q=q_{\alpha,\Sigma,V}$ defined by
\begin{equation} \label{Hamform}
q[\psi] := \|\nabla\psi\|^2 + (\psi, V\psi) -\alpha \int_{\mathbb{R}^2} |\psi(\Sigma(s))|^2\,g^{1/2}(s)\, \mathrm{d}s_1\, \mathrm{d}s_2\,,
\end{equation}
where $s=(s_1,s_2)$ are the coordinates used to parametrize $\Sigma$
and $g=\det(g_{ij})$ is the appropriate squared Jacobian defined by
means of the metric tensor $(g_{ij})$, with the domain
$H^1(\mathbb{R}^3)$, is closed and below bounded. Thus it is
associated with a unique self-adjoint operator which we identify
with $H= H_{\alpha,\Sigma,V}$ of \eqref{hamiltonian} above. In fact
such a claim is valid for a much wider class of potentials, however,
we focus here our attention on a particular case. By hypothesis (a)
above the surface $\Sigma$ splits $\mathbb{R}^3$ into two regions,
and we assume that
  \begin{enumerate}
  \item[(e)] $\,V(x)=V_0 >0$ in one of these regions and $V(x)= 0$ in the other.
    \end{enumerate}

\subsection{An auxiliary problem}

In the trivial case we have excluded the problem is solved easily by
separation of variables. It is useful to look at the transverse part
which we will need in the following. It is given by the operator
\begin{equation} \label{trans-op}
h = -\frac{\mathrm{d}^2}{\mathrm{d} x^2} -\alpha \delta (x) + V(x)\,,
\end{equation}
where $V(x) = V_0$ for $x > 0$ and $V(x) = 0$ otherwise, associated
with the quadratic form $\phi \mapsto \|\phi'\|^2 -\alpha|\phi(0)|^2
+ (\phi, V\phi)$ defined on $H^1(\mathbb{R})$. Properties of this
operator are easily found; we adopt without proof from \cite{EV16}
the following simple results.
\begin{lemma}
 \label{prop-h}
  \begin{enumerate}[(i)]
  \setlength{\itemsep}{2pt}
  \item $\sigma_\mathrm{ess} (h) = [0, \infty )$.
  \item The operator $h$ has no eigenvalues for $V_0 \ge \alpha^2$.
  \item The operator $h$ has a unique eigenvalue $\mu = - \left( \frac{ \alpha^2 - V_0}{2\alpha}\right)^2$ for $V_0 < \alpha^2$.
  \item If $V_0 = \alpha^2$ the equation $h\psi = 0$ has a bounded weak solution $\psi \notin L^2(\mathbb{R})$.
  \item For $V_0 > 0$ and any $\varphi \in C^2(\mathbb{R}_+) \bigcap L^2(\mathbb{R}_+)$ we have
$$
\int_0^\infty ({|\varphi^{\prime}}|^2 + V_0|\varphi|^2)(x)\, \mathrm{d} x \ge \sqrt{V_0}\, |\varphi(0)|^2.
$$
    \end{enumerate}
\end{lemma}
Relations between the coupling constant and the potential bias will
play an important role in the following. For the sake of brevity, we
shall call the case (iv) of the lemma \emph{critical}, and similarly
we shall use the terms \emph{subcritical} for case (iii) and
\emph{supercritical} for the situation where $V_0>\alpha^2$.

\subsection{The results}

Let us look first at the essential spectrum. As usual in the
Schr\"odinger operator theory it is determined by the behavior of
the interaction at large distances. In view of assumption (c) we
expect that asymptotically the situation approaches the trivial case
with separated variables mentioned above, and indeed, we have the
following result.
\begin{theorem} \label{spess}
$\sigma_\mathrm{ess} (H) = [\mu,\infty)\,$ holds under the assumptions (a)--(e), where $\mu := -\frac14 \alpha ^{-2}(\alpha ^2 -V_0)^2$ for $V_0 < \alpha ^2$ and $\mu := 0$ otherwise.
\end{theorem}

The question about the existence of the discrete spectrum is more
involved and the potential bias makes the answer distinctively
asymmetric. In particular, a critical or supercritical potential
supported in the \emph{exterior} region prevents negative
eigenvalues from existence.
\begin{theorem} \label{ext-abs}
Under the stated assumptions, suppose that $V(x)=0$ holds in
$\Omega_\mathrm{int}$ and $V(x)=V_0\ge \alpha^2$ in
$\Omega_\mathrm{ext}$, then $\sigma(H) = \sigma_\mathrm{ess}(H) =
[0,\infty)\,$.
\end{theorem}

On the other hand, the operator \eqref{hamiltonian} may have
isolated eigenvalues in the (sub)critical regime as we are going to
illustrate on examples. In Section~\ref{s:cone} we discuss the case
of a conical surface and show that the discrete spectrum of $H$ is
nonempty provided $V_0$ is small enough. The most interesting,
though, is the \emph{critical} case, $V_0=\alpha^2$ with the bias in
the \emph{interior} region. In Section~\ref{s:rooftop} we discuss
another example, this time with $\Sigma$ being a `rooftop' surface,
and show that for suitable values of parameters we have here
$\sigma_\mathrm{disc}(H)\ne\emptyset$.

\section{Proof of Theorem~2.2}

The argument splits into two parts. First we shall demonstrate the
implication
\begin{equation} \label{spess1}
\nu \in \sigma_\mathrm{ess}(H) \quad\text{if}\;\;\nu \ge \mu\,.
\end{equation}
To this goal, one has to find for any fixed $\nu \ge \mu$ and any
$\varepsilon > 0$ an infinite-dimensional subspace $\mathcal{L}
\subset \mathrm{Dom}(H)$ such that $\|(H-\nu) \psi \| < \varepsilon$
holds for any $\psi \in \mathcal{L}$. We denote $\zeta :=\nu-\mu \ge
0 $ and choose a pair of functions of unit $L^2$ norm, $f \in
C_0^2(\mathbb{R}^2)$ with the property that $\big\| (-\Delta -
\zeta)f\big\| < \frac14 \varepsilon$, and $g \in C_0^2(\mathbb{R})$
such that $\|g\|=1$ and $\| (h - \mu)g\| < \frac14 \varepsilon$.
Such functions can always be found in view of the fact that the
essential spectrum of the two-dimensional Laplacian is $[0,\infty)$
and of the properties of the operator $h$ stated in
Lemma~\ref{prop-h}.

In the next step we choose a sequence of surface points
$a_j=\Sigma(s^{(j)})$ such that $|s^{(j)}|\to\infty$ as
$j\to\infty$. With each of them we associate the Cartesian system of
coordinates $(x^{(j)},y^{(j)})$ where $x^{(j)}$ are the Cartesian
coordinates in the tangential plane to $\Sigma$ at the point $a_j$
and $y^{(j)}$ is the distance from this tangential plane. By
assumption (b), the points $a_j$ can be always chosen in such a way
that this coordinate choice makes sense. This allows us to define
the functions $\psi_j(x^{(j)},y^{(j)}) = f(x^{(j)})g(y^{(j)})$. By
construction, each of them has a compact support of the diameter
independent of $j$, hence in view of assumption (d) one can pick
them so that $\mathrm{supp}\,\psi_j \cap \Sigma$ is simply
connected. Using then a straightforward telescopic estimate in
combination with the requirement $\|f\|=\|g\|=1$ we get the
inequality
\begin{equation} \label{HVZeq1}
\|(H-\nu)\psi_j\| \leq \Big\|(-\Delta x -\zeta)f\Big\|  +\|(h-\mu)g\| + V_0 \|\psi_j|_{\mathcal{A}_j}\| +\|(\delta_{\Sigma_j} -\delta_\Sigma) \psi_j\|\,,
\end{equation}
where $\Sigma_j$ is the tangential plane at $a_j$ and
$\mathcal{A}_j$ is the part of the function support squeezed between
$\Sigma$ and $\Sigma_j$;  the last term is understood as the
$L^2$-norm over the two surface segments contained in the border of
$\mathcal{A}_j$. The first two terms on the right-hand side of
\eqref{HVZeq1} are by construction bound by $\frac12 \varepsilon$.
Furthermore, in view of the assumptions (a)--(c) in combination with
the smoothness of the functions $f,g$, which are the same for all
the $\psi_j$, the other two terms tend to zero as $j\to \infty$,
hence $\|(H-\nu)\psi_j\| < \varepsilon$ holds for all $j$ large
enough. In addition, one can always choose the points $a_j$ in such
a way that $\mathrm{supp}\,\psi_j \cap \mathrm{supp}\,\psi_{j'} =
\emptyset$ holds for $j \neq j'$, which means that Weyl's criterion
hypothesis is satisfied.

To complete the proof of the theorem, we have to demonstrate the
opposite implication, in other words, to check the validity of the
relation
\begin{equation} \label{spess2}
\sigma_\mathrm{ess}(H) \cap (-\infty, \mu) = \emptyset\,,
\end{equation}
which is equivalent to $\inf \sigma_\mathrm{ess}(H) \ge \mu$. While
in the first part of the proof we have extended to the present case
the argument used in the two-dimensional situation, now we choose a
different approach because the localization estimates employed in
\cite{EV16} become more involved here. We will need an auxiliary
result which is a sort of modification of Proposition~2.5 in
\cite{EY02}.

\begin{lemma}
 \label{trans-est}
Let $h_N$ denote the operator \eqref{trans-op} acting on the
interval $(-d,d)$ with Neumann boundary conditions at the endpoints,
associated with the form $\phi \mapsto \|\phi'\|^2
-\alpha|\phi(0)|^2  + (\phi, V\phi)$ defined on $H^1(-d,d)$. If
$V_0\le\alpha^2$, there are positive $c_0, d_0$ such that for all
$d>d_0$ we have $\inf\sigma(h_N) \ge \mu - c_0 d^{-1}$. If, on the
other hand, $V_0>\alpha^2$ holds, then $h_N\ge 0$ for all $d$ large
enough.
\end{lemma}
\begin{proof}
We observe that the ground-state eigenfunction of $h_N$
corresponding the eigenvalue $\mu_d<0$ is of the form
$$
\psi(x) = c_1 \chi_{(0,d)} \cosh \kappa_1(x-d) + c_2 \chi_{(-d,0)} \cosh \kappa_1(x+d)\,,
$$
where $\kappa_1:= \sqrt{-\mu_d}$ and $\kappa_2:= \sqrt{V_0-\mu_d}$.
Since the function has to be continuous at $x=0$ and satisfy
$\psi'(0+) - \psi'(0-)= -\alpha\psi(0)$, we get the spectral
condition
\begin{equation} \label{trans-neg}
\kappa_1 \tanh\kappa_1 d + \kappa_2 \tanh\kappa_2 d = \alpha\,.
\end{equation}
As a function of $-\mu_d$, the left-hand side is increasing from
$\sqrt{V_0} \tanh \sqrt{V_0}d$, behaving asymptotically as
$2\sqrt{-\mu} + \mathcal{O}((-\mu)^{-1/2})$, the equation
\eqref{trans-neg} has a unique solution for any fixed $d>0$ provided
$V_0\le\alpha^2$. Since the left-hand side is monotonous also with
respect to $d$ we have $\mu_d<\mu_\infty$ where $\mu_\infty=\mu$ of
Lemma~\ref{prop-h}(iii).

To get a lower bound we have to estimate the left-hand side of
\eqref{trans-neg} from below. We will do that using the rough bound
$\tanh x > 1-2\mathrm{e}^{-2x} > 1-x^{-1}$. Writing the solution of
the appropriate estimating condition in the form $\tilde\mu_d =
\mu-\delta$ we find after a short computation that $\tilde\mu_d =
\mu -2d^{-1} + \mathcal{O}(d^{-2})$ holds as $d\to\infty$ which
together with the inequality $\mu_d>\tilde\mu_d$ yields the result.

If $V_0>\alpha^2$ no solution exists for a sufficiently small $d$.
The condition \eqref{trans-neg} can be then modified replacing one
or both hyperbolic tangents by the trigonometric one, however, we
will not need it; it is enough to note that the lowest eigenvalue of
$h_N$ --- which certainly exists as $h_N$ as a Sturm-Liouville
operator on a finite interval has a purely discrete spectrum --- is
positive for $d$ large enough.
\end{proof}

Consider now the neighborhood $\Omega_d:= \{ x\in\mathbb{R}^3:\:
\mathrm{dist}(x,\Sigma)<d\}$ of the surface. Furthermore, fix a
point $x_0\in \Sigma$ and divide $\Sigma$ into two parts, $\Sigma_R$
consisting of the points the geodesic distance from $x_0$ is larger
than $R$ and $\Sigma_R^c= \Sigma \setminus \Sigma_R$. Note that by
assumptions (a) and (b) outside a compact $\Sigma$ has a well
defined normal and the points of $\Omega_d$ can be written as
$x_\Sigma + n_{x_\Sigma}u$ with $|u|<d$, where $x_\Sigma$ is the
point satisfying $\mathrm{dist}(x,\Sigma) =
\mathrm{dist}(x,x_\Sigma)$; for $d$ small enough this part of
$\Omega_d$ does not intersect itself, i.e. the point $x_\Sigma$ is
unique. Consequently, for $R$ sufficiently large and $d$
sufficiently small we may define $\Omega_{d,R}:= \{ x_\Sigma +
n_{x_\Sigma}u:\: x_\Sigma \in \Sigma_R,\: |u|<d\}$ and
$\Omega_{d,R}^c:= \Omega_d \setminus \Omega_{d,R}$. We also
introduce the $\Omega_d^c:= \mathbb{R}^3 \setminus \Omega_d$
consisting of two connected components, so together we have
$\mathbb{R}^3 = \Omega_{d,R} \cap \Omega_{d,R}^c \cap \Omega_d^c$.

Using these notions we employ a bracketing argument. Changing the
domain of $H$ by additional Neumann conditions imposed at the
boundaries of the three domains we obtain a lower bound to our
operator,
$$
H \ge H_{\Omega_{d,R}} \oplus H_{\Omega_{d,R}^c} \oplus H_{\Omega_d^c}\,.
$$
For the proof of \eqref{spess2} only the first part on right-hand
side is relevant, because $H_{\Omega_{d,R}^c}$ corresponds to a
compact region and its essential spectrum is thus void, and $\inf
\sigma_\mathrm{ess} (H_{\Omega_d^c})=0$ holds obviously. To analyze
the first part, which we for brevity denote as $H_{d,R}$ we employ a
geometric argument similar to that used in \cite{EK03}, the
difference being the constant potential $V_0$ to one side of
$\Sigma_R$. The `pierced layer' $\Omega_{d,R}$ can be regarded as a
submanifold in $\mathbb{R}^{3}$ equipped with the metric tensor
\begin{equation}\label{metricG}
G_{ij}=
\left( \begin{array}{cc} (G_{\mu\nu}) & 0 \\
0 & 1 \end{array} \right)\,, \quad G_{\mu\nu}=(\delta
_{\mu}^{\sigma}-uh_{\mu}^{\;\sigma}) (\delta _{\sigma}^{\rho}
-uh_{\sigma}^{\:\rho})g_{\rho\nu}\,,
\end{equation}
referring to the curvilinear coordinates $(s_1,s_2,u)$, where
$g_{\rho\nu}$ is the metric tensor of $\Sigma_R$ and
$h_{\mu}^{\;\sigma}$ is the corresponding Weingarten tensor; we
conventionally use the Greek notation for the range $(1,2)$ of the
indices and the Latin for $(1,2,3)$ and we employ the Einstein
summation convention. In particular, the volume element of $\Omega
_{d}$ is given by $d\Omega := G^{1/2}\mathrm{d}^{2}s\,\mathrm{d}u$
with
\begin{equation}\label{Jacobi}
G:=\det(G_{ij} =g\left[ (1-uk_1)(1-uk_2)\right]^2=g(1-2Mu+Ku^2)^2\,;
\end{equation}
for brevity we use the shorthand $\xi(s,u)\equiv 1-2M(s)u+K(s)u^2$. Next we introduce
$$
\varrho_R:= (\{\max_{\Sigma_R}\left\| k_{1}\right\|_{\infty }, \left\|
k_{2}\right\|_{\infty}\})^{-1}\,.
$$
By assumption (c) we have $\varrho_R\to\infty$ as $R\to\infty$, and
as self-intersections of $\Omega_{d,R}$ are avoided as long
$d<\varrho_R$; we see that the layer halfwidth $d$ can be in fact
chosen arbitrarily big provided $R$ is sufficiently large. At the
same time, the transverse component of the Jacobian satifies the
inequalities $C_-(d,R)\leq \xi(s,u) \leq C_+(d,R)$, where
$C_{\pm}(d,R):=(1\pm d\varrho_R^{-1})^2$, hence for a fixed $d$ and
$R$ large the Jacobian is essentially given by the surface part;
recall that the metric tensor $g_{\mu\nu}$ is assumed to be unifomly
elliptic, $c-\delta_{\mu\nu}\leq g_{\mu\nu}\leq c_+\delta_{\mu\nu}$
with positive $c_\pm$.

Now we use these geometric notions to asses the spectral threshold
of $H_{d,R}$. Passing to the curvilinear coordinates $(s,u)$, we
write the corresponding quadratic form as
$$
(\partial_i\psi,G^{ij}\partial_j\psi )_G +(\psi, V\psi)_G -\alpha
\int_{|s|>R}\left| \psi(s,0)\right|^2\, \mathrm{d}\Sigma
$$
defined on $H_1(\Omega_{d,R}^\mathrm{flat}, \mathrm{d}\Omega )$,
where $(\cdot,\cdot)_G$ means the scalar product in $
L^2(\Omega_{d,R}^\mathrm{flat}, \mathrm{d}\Omega)$, the symbol
$\mathrm{d}\Sigma$ stands for $g^{1/2}(s)\,\mathrm{d}s$ and
$\Omega_{d,R}^\mathrm{flat} :=\{q:\: |s|>R,\,|u|<d \}$. Using the
diagonal form \eqref{metricG} of the metric tensor together with
\eqref{Jacobi} and the above mentioned bound to the factor
$\xi(s,u)$ we find
\begin{eqnarray*}
\lefteqn{(\psi, H_{d,R}\psi)_G \ge \int_{\Omega_{d,R}^\mathrm{flat}} \left( |\partial_u \psi|^2 +V|\psi|^2 \right) \mathrm{d}\Omega - \alpha
\int_{|s|>R}\left| \psi(s,0)\right|^2\, \mathrm{d}\Sigma} \\ &&
\ge \xi_R^- \int_{\Omega_{d,R}^\mathrm{flat}} \left( |\partial_u \psi|^2 +V|\psi|^2 \right) \mathrm{d}\Sigma\, \mathrm{d}u - \alpha
\int_{|s|>R}\left| \psi(s,0)\right|^2\, \mathrm{d}\Sigma\,,
\end{eqnarray*}
where $\xi_R^- := \inf_{\Omega_{d,R}} \xi(s,u)$. Introducing
similarly $\xi_R^+ := \sup_{\Omega_{d,R}} \xi(s,u)$ and using
Lemma~\ref{trans-est} with the coupling constant $\alpha_R =
\frac{\alpha}{\xi_R^-}$, we get
$$
(\psi, H_{d,R}\psi)_G \ge \frac{\xi_R^-}{\xi_R^+}\left(\mu - c_0 d^{-1} \right) \|\psi\|_G^2
$$
provided $d$ is large enough. It is clear from the above discussion
that to any $\varepsilon>0$ one can choose $R$ and $d$ sufficiently
large to get $c_0 d^{-1} < \frac12\varepsilon$, and at the same
time, $\varrho_R>d$ and $\frac{\xi_R^-}{\xi_R^+} <
\frac{\varepsilon}{2\mu}$. Consequently, $\inf
\sigma_\mathrm{ess}(H) > \mu - \varepsilon$ which concludes the
proof.

\begin{remarks}
(a) The proof did not employ the part of assumption (a) speaking about convexity of one of the regions to which the surface divides the space and, in fact, neither the fact that $\Sigma$ is simply connected. \\
(b) The bound in Lemma~\ref{trans-est} is a rough one but it suffices for the present purpose; in reality the error caused by the Neumann boundary is exponentially small as $d\to\infty$, similarly as in \cite{EY02}. Note that we use such an estimate in a different way than in the said work: there we made the error small by choosing a large $\alpha$ while here the coupling constant is fixed but we choose a large $d$ which we are allowed to do being far enough in the asymptotic region.
\end{remarks}


\section{Proof of Theorem~2.3}

In view of Theorem~\ref{spess} it is sufficient to check that
$q[\psi]\ge 0$ holds for any $\psi \in C_0^2(\mathbb{R}^3)$. The
contribution from $\Omega_\mathrm{int}$ to the quadratic form is
non-negative and may be neglected; this yields the estimate $q[\psi]
\ge q_\mathrm{ext}[\psi]$, where
\begin{equation}\label{Jacobi}
q_\mathrm{ext}[\psi]:= \int_{\Omega_\mathrm{ext}} \left(|\nabla \psi |^2 + V_0 |\psi|^2 \right)(x)\, \mathrm{d}x -\alpha \int_{\mathbb{R}^2} |\psi(\Sigma(s))|^2\,g^{1/2}(s)\, \mathrm{d}^2s\,.
\end{equation}
To estimate the quantity from below, we note that by assumption (b)
there is a family of open connected subsets
$\Sigma_l,\:l=1,\dots,N$, of $\Sigma$ which are mutually disjoint
and such that $\Sigma = \overline{\cup_l \Sigma_l}$ and
$\Sigma|_{\Sigma_l}$ is $C^2$ smooth for any $l$. It may happen that
$N=1$ if $\Sigma \setminus (\mathcal{C} \cup \mathcal{P})$ is
connected, on the other hand, $N\ge 2$ has to hold, for instance,
when one of the curves of the family $\mathcal{C}$ is closed, or
more generally, if $\mathcal{C}$ contains a loop. The number of the
$\Sigma_l$'s can be made larger if we divide a smooth part of the
surface by an additional boundary, but by assumption (b) the
partition can be always chosen to have a finite number of elements,
and moreover, only one of the $\Sigma_l$'s is not precompact in
$\Sigma$.  Next we associate the sets $\Omega_l:= \{ x_\Sigma +
n_{x_\Sigma}u:\: x_\Sigma \in \Sigma_l,\: u<0\} \subset
\Omega_\mathrm{ext}$ with the $\Sigma_l$'s, where the negative sign
refers to the fact that the normal vector points conventionally into
the interior domain. Since the latter is convex by assumption~(a),
no two halflines $\{x_\Sigma \in \Sigma_l,\: u<0\}$ emerging from
different points of $\Sigma$ can intersect, and consequently, the
sets $\Omega_l$ are mutually disjoint, if $\mathcal{C} \cup
\mathcal{P} \ne \emptyset$ the closure $\overline{\cup_l \Omega_l}$
may be a proper subset of the exterior domain. This yields
$$
q_\mathrm{ext}[\psi] \ge \sum_l \int_{\Omega_l} \left(|\nabla \psi |^2 + V_0 |\psi|^2 \right)(x)\, \mathrm{d}x -\alpha \int_{\mathbb{R}^2} |\psi(\Sigma(s))|^2\,g^{1/2}(s)\, \mathrm{d}^2s
$$
and passing in the first integral to the curvilinear coordinates in
analogy with the previous proof we get
\begin{eqnarray*}
\lefteqn{q_\mathrm{ext}[\psi] \ge \sum_l \int_{M_l} \int_{-\infty}^0 \big(((\overline{\partial_i\psi})\, G^{ij} (\partial_j\psi)+ V_0 |\psi|^2)\, G^{1/2}\big) (s,u)\: \mathrm{d}^2s\, \mathrm{d}u} \\ && \qquad -\alpha \int_{\mathbb{R}^2} |\psi(\Sigma(s))|^2\,g^{1/2}(s)\, \mathrm{d}^2s\,, \phantom{AAAAAAAAAAAAAAAAA}
\end{eqnarray*}
where $M_l$ is the pull-back of the surface component $\Sigma_l$ by
the map $\Sigma$. Neglecting the non-negative term
$(\overline{\partial_\mu\psi})\, G^{\mu\nu} (\partial_\nu\psi)$ and
using (\ref{Jacobi}) we arrive at the estimate
\begin{eqnarray*}
\lefteqn{q_\mathrm{ext}[\psi] \ge \sum_l \int_{M_l} \int_{-\infty}^0
(|\partial_u\psi|^2+ V_0 |\psi|^2)(s,u)\, g^{1/2}(s)\, (1-uk_1(s))}
\\ && \qquad \times(1-uk_2(s)) \mathrm{d}^2s\, \mathrm{d}u -\alpha
\int_{\mathbb{R}^2} |\psi(\Sigma(s))|^2\,g^{1/2}(s)\, \mathrm{d}^2s
\,. \phantom{AAAA}
\end{eqnarray*}
However, $1-uk_\mu(s)\ge 1$ holds for $\mu=1,2$ and $u<0$ because
both the principal curvatures are non-negative in view of the
convexity assumption. Furthermore, the difference between $\cup_l
M_l$ and $\mathbb{R}^2$ is a zero measure set, hence we finally find
$$
q[\psi] \ge \int_{\mathbb{R}^2} \left\{\int_{-\infty}^0 (|\partial_u\psi|^2+ V_0 |\psi|^2)(s,u)\, \mathrm{d}u - \alpha |\psi(s,0)|^2    \right\}\, g^{1/2}(s)\, \mathrm{d}^2s \,,
$$
where, with the abuse of notation, we have employed the symbol
$\psi(s,0)$ for $\psi(\Sigma(s))$, but $\alpha\le\sqrt{V_0}$ holds
by assumption, and consequently, the expression in the curly
brackets is positive by Lemma~\ref{prop-h}(v). This concludes the
proof.

\section{Example of a conical surface} \label{s:cone}

Consider now the the situation where $\Sigma=\mathcal{C}_\theta$ is
a circular conical surface of an opening angle $2\theta \in
(0,\pi)$, in other words
$$
\mathcal{C}_\theta = \big\{ (x,y,z)\in\mathbb{R}^3:\: z= \cot\theta\, \sqrt{x^2+y^2} \big\}\,.
$$

This surface satisfies the assumptions (a)--(e) with $\mathcal{P}$
consisting of a single point, the tip of the cone, hence if the
potential bias is supported in the exterior of this cone, the
spectrum of the corresponding operator is by Theorems~\ref{spess}
and \ref{ext-abs} purely essential, $\sigma(H)=[\mu,\infty)$. Let us
look now what happens in the opposite case when the bias is in the
interior.

In view of the symmetry it is useful to employ the cylindrical
coordinates relative to the axis of $\mathcal{C}_\theta$. Since our
potential is independent of the azimuthal angle, the operator $H$
now commutes with the corresponding component of the angular
momentum operator, $-i\partial_\varphi$, and allows thus for a
partial-wave decomposition, $H= \bigoplus_{m\in\mathbb{Z}} H^{(m)}$.
Writing the wave function in the standard way through its reduced
components,
$$
\psi(r,\varphi,z) = \sum_{m\in\mathbb{Z}} \frac{\omega_m(r,z)}{\sqrt{2\pi r}}\, \mathrm{e}^{im\varphi}\,,
$$
we can rewrite the quadratic form (\ref{Hamform}) as the sum $q[\psi] = \sum_{m\in\mathbb{Z}} q_m[\psi]$, where
\begin{eqnarray} \label{redform}
\lefteqn{q_m[\psi] := \|\nabla\omega_m\|^2_{L^2(\mathbb{R}^2_+)} +
\int_{\mathbb{R}^2_+} \frac{4m^2-1}{4r^2}\, |\omega_m(r,z)|^2\, \mathrm{d}r\,\mathrm{d}z } \\ && \quad +
\int_{\mathbb{R}^2_+} V(r,z)\, |\omega_m(r,z)|^2\, \mathrm{d}r\,\mathrm{d}z - \alpha^2 \|\omega_m|_{\Gamma_\theta}\|_{L^2(\mathbb{R}^2_+)}\,, \nonumber
\end{eqnarray}
where $\mathbb{R}^2_+$ is the halfplane $\{(r,z):\: r>0,\,
z\in\mathbb{R}\}$ and $\Gamma_\theta$ is the halfline
$z=r\,\cot\theta$. We note first that it is sufficient to focus on
the component with vanishing angular momentum in the partial-wave
decomposition.

\begin{proposition} \label{nonzero_m}
Let $V(x)=V_0\,\chi_{\Omega_\mathrm{int}}(x)$ with $V_0>0$, then $q_m[\psi]\ge\mu$ holds for any nonzero $m\in\mathbb{Z}$ and all $\psi\in H^1(\mathbb{R}^3)$.
\end{proposition}
\begin{proof}
If $m\ne 0$, the second term on the right-hand side of (\ref{redform}) is non-negative and one estimate the form from below by neglecting it. Following the paper \cite{BEL14} where the case $V_0=0$ is treated, we introduce in the halfplane $\mathbb{R}^2_+$ another orthogonal coordinates, $s$ measured along $\Gamma_\theta$ and $t$ in the perpendicular direction; the axes of the $(t,s)$ system are rotated with respect to those of $(r,z)$ around the point $r=z=0$ by the angle $\theta$. In these coordinates we have
$$
q_m[\psi] \ge \|\nabla\omega_m\|^2_{L^2(\mathbb{R}^2_+)} +
V_0 \int_0^\infty \mathrm{d}s \int_{-s\tan\theta}^\infty |\omega_m(s,t)|^2\, \mathrm{d}t - \alpha \int_0^\infty |\omega_m(s,0)|^2\, \mathrm{d}s\,,
$$
where with the abuse of notation we write $\omega_m(s,t)$ for the
function value in the rotated coordinates. If $\psi\in
H^1(\mathbb{R}^3)$ and $m\ne 0$, the reduced wave function
$\omega_m$ belongs to $H^1(\mathbb{R}^2_+)$. By $\tilde\omega_m$ we
denote its extension to the whole plane by the zero value in the
other halfplane, $\{(s,t):\: s\in\mathbb{R},\: t<-s\tan\theta\}$,
which naturally belongs to $H^1(\mathbb{R}^2)$. Such function form a
subspace in $H^1(\mathbb{R}^2)$, however, hence we have
\begin{eqnarray*}
\lefteqn{\inf_{\omega_m\in H^1(\mathbb{R}^2_+)} q_m[\psi] = \inf_{\omega_m\in H^1(\mathbb{R}^2_+)} \bigg\{ \|\nabla\tilde\omega_m\|^2_{L^2(\mathbb{R}^2_+)} +
V_0 \int_{\mathbb{R}^2} |\tilde\omega_m(s,t)|^2\, \mathrm{d}s\, \mathrm{d}t} \\ && \hspace{12em} - \alpha \int_{\mathbb{R}} |\tilde\omega_m(s,0)|^2\, \mathrm{d}s \bigg\} \\ && \ge \inf_{\varrho\in H^1(\mathbb{R}^2)} \bigg\{ \|\nabla\varrho\|^2_{L^2(\mathbb{R}^2_+)} +
V_0 \int_{\mathbb{R}^2} |\varrho(s,t)|^2\, \mathrm{d}s\, \mathrm{d}t - \alpha \int_{\mathbb{R}} |\varrho(s,0)|^2\, \mathrm{d}s \bigg\}\,.
\end{eqnarray*}
Noting finally that form on the right-hand side of the last
inequality is associated with the self-adjoint operator in
$L^2(\mathbb{R}^2)$ which has separated variables,
$\overline{-\partial^2_s \otimes I_t + I_s \otimes h}$, where $h$ is
the operator (\ref{trans-op}) with the variable $x$ replaced by
$-t$, we obtain the desired claim from Lemma~\ref{prop-h}.
\end{proof}

On the other hand, the component with zero momentum can give rise to
a nontrivial discrete spectrum, at least as long the potential bias
is weak enough.

\begin{proposition} \label{weak bias}
To any integer $N$ there is a number $v_N \in (0,\alpha^2)$ such
that $\#\sigma_\mathrm{disc}(H)\ge N$ holds for $0\le V_0 < v_N$.
\end{proposition}
 \begin{proof}
The potential bias we consider is a bounded perturbation, therefore
$\{H_{\alpha,\mathcal{C}_\theta,V}:\: V_0\ge 0\}$ is a type (A)
holomorphic family in the sense of \cite{Ka}. This means, in
particular, that the eigenvalues, if they exist, are continuous
functions of $V_0$. The same is by Theorem~\ref{spess} true for the
essential spectrum threshold. Since the bias-free operator
$H_{\alpha,\mathcal{C}_\theta,0}$ has by \cite{BEL14} an infinite
number of isolated eigenvalues accumulating at $-\frac14\alpha^2$,
the continuity implies the result.
 \end{proof}

Before proceeding further, let us note than in the
higher-dimensional analogue of this problem the geometrically
induced discrete spectrum is void in the absence of the bias, cf.
\cite{LO16}, hence one does not expect a counterpart of
Proposition~\ref{weak bias} to hold either.

\section{Example of a rooftop surface} \label{s:rooftop}

In the previous example we left open the question whether the
discrete spectrum could survive up to the critical value of the
potential. To demonstrate that this is possible, consider now
another example in which the surface $\Sigma =
\mathcal{R}_{L,\theta}$ is defined through its cuts $\Gamma_z$ at
the fixed value of the coordinate $z$, i.e. $\Sigma = \{\Gamma_z:\:
z\ge 0\}$. We suppose that each $\Gamma_z$ is a $C^\infty$ loop in
the $(x,y)$-plane, being a border of a convex region, and consisting
of
\begin{enumerate}[(i)]
\item two line segments $\{(x,\pm z\tan\theta):\: |x|\le
\frac12 L\}$, and
\item two arcs connecting the loose `right' and `left' ends of the
segment, respectively. We suppose in addition that that these arcs
corresponding to different values of $z$ are mutually homothetic.
\end{enumerate}
We can regard $\mathcal{R}_{L,\theta}$ as coming from cutting the
cone of the previous example into two halves and inserting in
between a wedge-shaped strip of height $L$, modulo a smoothing in
the vicinity of the interface lines. It is easy to check that such a
surface satisfies assumptions (a)--(d) of Section~\ref{s:results}
with the set $\mathcal{C}$ consisting of the segment $\{(x,0,0):\:
|x|\le \frac12 L\}$ and $\mathcal{P}=\emptyset$. We have the
following result concerning the critical operator $H= H_{\alpha,
\mathcal{R}_{L,\theta}, \alpha^2}$:
\begin{proposition} \label{rooftop}
$\sigma_{disc}(H)$ is nonempty provided $L$ is sufficiently large,
and moreover, the number of negative eigenvalues can be made larger
than any fixed integer by choosing $\theta$ small enough.
\end{proposition}
 \begin{proof}
We employ the result from the two-dimensional case \cite{EV16},
where an attractive $\delta$ interaction supported by a broken line
of the opening angle $2\theta\in(0,\pi)$ gives rise at least one
bound states, and to a larger number for $\theta$ small. Let
$\phi=\phi(y,z)$ be an eigenfunction of the two-dimensional problem
corresponding to an eigenvalue $\lambda<0$. We choose a function
$g\in C_0^\infty(-1,1)$ and use
$$
\psi_\varepsilon:\: \psi_\varepsilon(x,y,z) = \varepsilon^{1/2}
g(\varepsilon x) \phi(y,z)
$$
as a trial function. Since the variables are separated, it is
straightforward to find the value of the quadratic form
\eqref{Hamform}, namely
$$
q[\psi_\varepsilon] = \big( \varepsilon\|g'\|^2 +\lambda \big)\,
\|\psi\|^2.
$$
The expression in the bracket can be made negative by choosing
$\varepsilon$ small enough, and since the support of
$\psi_\varepsilon$ lies within the layer $\{ (x,y,z):\: |x|\le
\frac12 L\}$, it is sufficient to choose $L>2\varepsilon$. Moreover,
the argument applies to any eigenvalue $\lambda$ of the
two-dimensional problem, which concludes the argument.
 \end{proof}

Note that the result will not change if the surface $\Sigma$ is
deformed outside the support of the trial function which means, in
particular, that convexity assumption may be weakened.

\subsection*{Acknowledgments}
The research was supported by the Czech Science Foundation within
the project 17-01706S.

\bibliographystyle{amsalpha}

\end{document}